\documentclass[letterpaper]{article} 

\usepackage{booktabs}
\usepackage{balance} 
\usepackage[margin=1in]{geometry}
\usepackage{graphicx} 
\usepackage{dblfloatfix}
\usepackage{algorithm}
\usepackage{algorithmic}
\usepackage{comment}

\usepackage{amsmath,amssymb,amsfonts,enumitem,amsthm,physics,mathtools}
\usepackage{xcolor}
\usepackage{natbib}
\usepackage{hyperref}
\usepackage{newfloat}
\usepackage{listings}

\usepackage{pgfplots}

\newtheorem{theorem}{Theorem}[section]
\newtheorem{lemma}[theorem]{Lemma}

\newtheorem{definition}[theorem]{Definition}

\newtheorem{remark}{Remark}

\newcommand{\E}{\mathbb{E}}
\newcommand{\C}{\mathcal{C}}

\newcommand\dungnote[1]{\textcolor{blue}{***d\~ung: {#1}}}

\title{Differentially Private Graph Coloring}

\author {
    Michael Xie\thanks{University of Maryland, College Park, MD, USA} ,
    Jiayi Wu\footnotemark[1] ,
    D\~ung Nguyen\thanks{Haverford College, Haverford, PA, USA. Correspondence to: D\~ung Nguyen (dnguyen1@haverford.edu)} ,
    Aravind Srinivasan\footnotemark[1] 
  }
\date{}

\begin{document}
\maketitle

\begin{abstract}
  Differential Privacy is the gold standard in privacy-preserving data analysis.
  This paper addresses the challenge of computing an edge-differentially private vertex coloring.
  In this paper, we present two novel algorithms for this problem.
  Both algorithms begin by coloring each vertex uniformly at random from a fixed-size palette, and then apply the exponential mechanism to locally resample colors for either all vertices or a selected subset of vertices.

  Any non-trivial edge differentially private coloring of a graph needs to be defective, as a proper coloring exposes the non-existence of an edge between two vertices of the same color.
 A coloring is $k$-defective if each vertex shares its color with at most $k$ of its neighbors.
  Our goal is to design coloring algorithms that use the minimum number of colors, while achieving the smallest possible defect under the edge-differential privacy.
  Our first algorithm applies to $d$-inductive graphs with maximum degree $\Delta$.
  We show that it yields a \(3\epsilon\)-differentially private coloring with \(O(\frac{\log n}{\epsilon}+d)\) maximum defect, using a palette of size $\Theta(\frac{\Delta}{\log n}+\frac{1}{\epsilon})$.
  Our second algorithm utilizes noisy thresholding to guarantee \(O(\frac{\log n}{\epsilon})\) maximum defect, using a palette of size $\Theta(\frac{\Delta}{\log n}+\frac{1}{\epsilon})$, generalizing the results to all graphs rather than just $d$-inductive ones.
\end{abstract}


\section{Introduction}
\label{sec:intro}
Graph coloring is a fundamental problem in graph theory with wide-ranging applications in data mining and knowledge discovery\cite{chaitin1981register, hale1980frequency, garey1976application, marx2004graph}. 
Informally, given a graph, the vertex coloring problem asks for an assignment of colors to vertices such that adjacent vertices receive distinct colors, while minimizing the total number of colors used.
In its canonical form, graph coloring is a hard problem~\cite{karp2009reducibility, zuckerman2006linear}, even for simple graphs~\cite{garey1974some, appel1977every}.
Graph coloring has been extensively studied in various settings, notably for specific types of graphs~\cite{grotschel1984polynomial, lovasz1972normal, gavril1972algorithms, gavril1974intersection, henzinger2020explicit}, approximation coloring~\cite{wigderson1983improving, halldorsson1993still}, distributed computing~\cite{barenboim2013distributed, kuhn2006complexity}, and in scalable and high-performance settings~\cite{gebremedhin2000scalable, boman2005scalable}.
\smallskip

Many applications of graph coloring involve sensitive data encoded in the topological structure of the input. 
Vertices may represent individuals and edges between them represent relationships or interactions, such as social networks, communication patterns, scheduling conflicts, etc. 
Releasing a coloring of such graphs may leak private information about the presence or absence of relationships, which motivates the study of graph algorithms under differential privacy (DP).
\smallskip

Differential privacy (DP) provides a rigorous framework for publishing statistical information about a dataset while limiting the information revealed about any single individual~\cite{dwork:fttcs14}.
It states that the output of a (randomized) private algorithm (or mechanism) does not change much between two neighboring input datasets that are identical except for one element.
The neighboring relationship defines the privacy model, i.e., the unit of privacy protection in DP.
For graph algorithms under DP, there are two common privacy models: edge- and node-privacy~\cite{10.1007/978-3-642-36594-2_26, blocki2013differentially}. 
\smallskip

The former conceals the presence of individual edges and the latter conceals entire nodes and all their incident edges (making it much stricter and infeasible for our purposes, which aim to assign a color to each node).
Even under the edge privacy model, an exact coloring would reveal many data points, i.e., the existence of an edge, or the lack thereof, since no edge may exist between any two vertices of the same color and vice versa. 
Hence we will work with defective colorings, allowing adjacent vertices to share the same color up to some degree (formally defined in Definition~\ref{def:defective}).
\smallskip

The main difficulty is the global effect of the coloring problem under edge-DP.
A single edge that is added to a graph may invalidate the colors of its endpoints, though the recoloring of the endpoints may cascade to a recoloring of the whole graph.
It makes analyzing the sensitivity, the amount of change in the output when an element of the input is modified, challenging.
Another issue, which was stated earlier, is that the exact coloring will expose the privacy of edges and non-edges and therefore, any non-trivial coloring must be defective.
Hence, our problem is now a bi-criteria optimization: minimizing both the palette size and the amount of defectiveness.
The defectiveness is also challenging to optimize locally.
When a node commits to its final color, its unfinished adjacent nodes may later choose the same color and inflate the defectiveness of the finished node.
Bounding these collisions is not straightforward for an arbitrary graph, which is the ultimate goal of this work.
\smallskip

\textbf{Our Contributions.}
In this work, we study the problem of graph coloring under differential privacy and present new algorithms that achieve improved utility guarantees and practical results.
Our main contributions can be summarized as follows:
\begin{itemize}
    \item We propose two differentially private coloring algorithms under edge differential privacy that outputs a defective coloring.
    \item We show that the first algorithm uses at most $\Theta(\frac{\Delta}{\log n}+\frac{1}{\epsilon})$ colors and \(O(\frac{\log n}{\epsilon}+d)\) defectiveness with high probability for $d$-inductive graphs, and the second algorithm with similar palette size achieves $O(\log{n}/\epsilon)$ defectiveness for arbitrary graphs.
    \item We complement our theoretical results with experimental evaluations on both synthetic and real-world graphs, demonstrating that our algorithms outperform prior differentially private graph coloring methods in terms of defectiveness.
\end{itemize}

Both algorithms utilize two-phase coloring: an initial coloring assigned uniformly at random, followed by a recoloring step with color distributions carefully calculated from the neighborhood of each vertex.
The two algorithms differ in how they control the post-finalization color commitment.
The first algorithm resamples all vertices in the order of a degeneracy ordering--an ordering that limits the out-degrees of nodes from left to right.
This allows us to achieve bounded defectiveness on graphs that naturally satisfy such an ordering.
The second algorithm removes this assumption by resampling selectively, using a noise test to recolor vertices whose noisy defectiveness exceeds a threshold, such that the color collisions are controlled for highly defective nodes.

\subsection{Related Work}
It is a classical result that even determining the chromatic number is NP-hard~\cite{Karp1972}. Hence a long line of research aimed to design efficient algorithms that exploit graph structures or relax the optimality. Efficient polynomial-time algorithms exist for several important special graph classes: including bipartite, chordal, interval, bounded arboricity, and perfect graphs.~\cite{grotschelperfect,10.5555/984029,DBLP:journals/corr/abs-2002-10142} For general graphs, algorithms exist for using $\Delta+1$ colors~\cite{10.1093/comjnl/10.1.85}, and additional works using semidefinite programming with improved approximations~\cite{karger1998approximategraphcoloringsemidefinite}. Variations of graph coloring has also garnered much attention, such as defective coloring~\cite{Archdeacon1987ANO,10.5555/1857976.1857977}.

A substantial body of work has developed differentially private algorithms for fundamental graph problems.
These include degree and subgraph counting \cite{5360242}, cut and flow problems \cite{10.5555/1873601.1873691}.
A recent work has approached DP graph coloring using the multidimensional above threshold mechanism, resulting in a palette size that scales with arboricity ($\alpha$) instead of maximum degree~\cite{dhulipala_et_al:LIPIcs.ESA.2025.91}, but such an approach would be better suited for low arboricity graphs.
Most closely related to our work is the paper \emph{Private Graph Colouring with Limited Defectiveness}~\cite{christiansen2024private}, which studies graph coloring under differential privacy using defective colorings. However, their approach does not fully exploit or adapt to the underlying topology or structural properties of the input graph. We show a comparison of our methods (Algorithms~\ref{alg:dp-exp} and~\ref{alg:dp-exp-threshold}) with prior work. We summarize our bounds and those of prior work in Table~\ref{tab:prior}
\begin{table*}[t]
    \centering
    \caption{Comparison of Our Methods with Prior Work}
    \label{tab:prior}
\begin{tabular}{c|c|c|c}
    Method & Privacy & Defectiveness & Palette \\
    \hline
    Algorithm~\ref{alg:dp-exp} & $3\epsilon$ & $O(\epsilon^{-1}\log n+d)$ & $O(\Delta/\log n)$\\ 
    &&&\\
    Algorithm~\ref{alg:dp-exp-threshold} & $5\epsilon$ & $\max \left\{
\begin{aligned}
&5\log n + \frac{2\log n + \log \Delta}{\epsilon}  + \frac{2\log n}{\epsilon},\\
&(1+\tfrac{4}{\epsilon})\log n
+ \tfrac{2}{\epsilon}\log \Delta
\end{aligned}
\right\}
+ \log n$ & $O(\Delta/\log n)$\\
    &&&\\
    
    Dhulipala et al.\cite{dhulipala_et_al:LIPIcs.ESA.2025.91} & $\epsilon$ & $O(\epsilon^{-1}\log n)$ & $O(1 + \alpha\epsilon/ \log n)$ \\
    Christiansen et al.\cite{christiansen2024private} & $\epsilon$ & $O(\log n)$ & $O(\Delta/\log n+\epsilon^{-1})$\\
\end{tabular}
\end{table*}

\section{Preliminaries}
\label{sec:prelim}
\subsection{Privacy and Probability Tools}
We shall now introduce some of the tools that we shall use. First we shall show some common differentially private mechanisms.
\begin{definition}
    A mechanism $M: \mathcal{X} \rightarrow \mathcal{Y}$ is $(\epsilon,\delta)$-differentially private (DP) if for any two neighboring inputs $X_1\sim X_2$, and any measurable subset of the output space $S\subseteq \mathcal{Y}$, the following holds: $$\Pr[M(X_1)\in S]\leq e^\epsilon \Pr[M(X_2)\in S]+\delta$$

    We say that $M$ is $\epsilon$-differentially private if $\delta=0$.
\end{definition}

We consider the edge-DP model~\cite{karwa:tds13}, where $V$, the set of nodes (or vertices), is public and $E$, the set of edges, is kept private, and only known by the private mechanism.
In this model, $\mathcal{X}$ corresponds to the set graphs with $n$ nodes. 
More formally, two networks $G_1=(V_1,E_1),G_2=(V_2,E_2)$, are considered neighbors if $V_1=V_2$ and there exists an edge $e$ such that $E_1=E_2\cup \{e\}$ or $E_2=E_1\cup \{e\}$ (i.e. they differ in the existence of a single edge).
The other common privacy model in graphs is the node-privacy model~\cite{kasiviswanathan2013analyzing, blocki:itcs13}.
However, since our goal is to output the colors of the nodes, it is only feasible in the edge-privacy model.

The exponential mechanism is a commonly used $\epsilon$-private mechanism to select one output to optimize a utility score $u$, defined as follows:

\begin{definition}
[\textbf{Exponential mechanism}]
Given a utility function $u:\mathcal{X}^n\times \mathcal{R} \rightarrow \mathbb{R}$, let $\operatorname{GS}_u=\max_{r\in \mathcal{R}}\max_{x\sim x'} |u(x,r)-u(x',r)|$ be the global sensitivity of $u$. The exponential mechanism $M(x,u,\mathcal{R})$ outputs an element $r\in \mathcal{R}$ with probability $\propto \exp(\frac{\epsilon u(x,r)}{2\operatorname{GS}_u})$.
\end{definition}

\begin{lemma}
  \label{lemma:utility}
  \textbf{(Utility of Exponential Mechanism)}
Let X be a dataset, and \(OPT(X) = \max_{h\in H} s(X, h)\) be the score attained by the best object \(h\) with respect to the dataset \(X\). For a dataset \(X\), let \(H^* = h \in H : s(X, h) = OP T(X)\) be the set of objects which achieve this score. Then 
\( Pr[s(M(X)) \le OPT(X) -\frac{2GS_s}{\epsilon}(ln\frac{|H|}{|H^*|}+ t)] \le exp(-t) \)
\end{lemma}

Another popular tool we shall use is the Laplace mechanism, which injects Laplacian noise to output in forms of scalars or vectors:
\begin{definition}
[\textbf{Laplace mechanism}]
Let $f:\mathcal{X}^n \rightarrow \mathbb{R}^d$ be a function with global $\ell_1$-sensitivity $GS_f = \max_{x \sim x'} |f(x) - f(x')|_1$. The Laplace mechanism releases
$$
M(x) = f(x)+(Z_1,\ldots, Z_d),
$$
where $Z_i \sim \operatorname{Lap}(GS_f/\epsilon)$ are independent random variables drawn from the Laplace distribution with scale parameter $GS_f/\epsilon$. $M(x)$ is $\epsilon$-differentially private.
\end{definition}


For more information and bounds of the mechanisms, we invite the reader to consult \cite{dwork:fttcs14} for a more comprehensive introduction.

We will use the following Chernoff bound:
\begin{lemma}
    \label{lemma:chernoff}
    \textbf{(Chernoff Bound)}
    Let $X_1, ..., X_m$ be independent random variables s.t $0\le X_i\le 1$. Let $S$ denote their sum and $\mu = \mathbb{E}[S]$. Then, for any $\eta\ge 0$, $$ P(S\ge (1+\eta)\mu)\le e^{-\frac{\eta^2\mu}{2+\eta}}$$
\end{lemma}

We will also introduce the notion of a d-inductive graph.
\begin{definition}
A graph $G$ is called a $d$-inductive graph  when there exists an ordering $v_1, v_2, ..., v_n$ such that every $v_i\in V(G)$ has most $d$ neighbors in $\{v_1, v_2,..., v_{i-1}\}$.
\end{definition}

\subsection{Defective Coloring}

\begin{definition} \textbf{($(C, k)$-defectiveness)}
\label{def:defective}
A coloring is $(C, k)$-defective if it uses a color palette size of $C$ and any vertex may share the same color with at most $k$ of its adjacent vertices.
\end{definition}

Using Definition\ref{def:defective}, we may now define our problem:
\begin{definition}
    Given an input graph $G=(V,E)$, we design an $\epsilon$-edge-differentially private and $(C,k)$-defective algorithm that minimizes $C$ and $k$ (the palette size and the defectiveness respectively).
\end{definition}

We use $A$ to indicate the adjacency matrix of $G$.
$N(u)$ denotes the set of adjacent nodes to $u$: $N(u) = \{v : A(u,v) = 1\}$.
$deg_G(u) = |N(u)|$ indicates the degree of node $u$ in graph $G$.
$color_G(u)$ is the current assigned color of node $u$.
$\Delta_G$ denotes the maximum degree of nodes in graph $G$: $\Delta_G = max_{u\in V(G)}deg_G(u)$. 
Let $\C$ denote all possible coloring schemes of graphs of size $n$.
We drop the subscript $G$ when the context is clear.

\subsection{Notation}
Table~\ref{tab:cor} lists the notations we shall use in our paper.

\begin{table}
\caption{Notation Used}
\label{tab:cor}
    \centering
\begin{tabular}{c|c}
    Notation & Explanation \\
    \hline
    $G(V,E)$ & Graph (Vertex, Edge sets) \\
    $n$ & Number of vertices ($|V|$) \\
    $\Delta_G$ & Max degree \\
    $C$ & Palette size\\
    $k$ & Defectiveness
\end{tabular}

\end{table}

\section{Private defective coloring}
\label{sec:approach}

\subsection{Iterative sampling algorithm}

\begin{algorithm}
  \caption{Private coloring algorithm via the Exponential mechanism $M_{Unctr}$\\
    \textbf{Input: } Graph $G=(V,E)$, privacy parameters $\epsilon$\\
    \textbf{Output: } Coloring $c:V\to [C]$}
    \begin{algorithmic}[1]
        \STATE $\tilde{\Delta}\gets\Delta_G+\text{Lap}(\frac{1}{\epsilon})$
    	\STATE $C\gets \frac{\tilde{\Delta}}{\log |V|}$
        \STATE $c \gets \{\, v \mapsto \text{Uniform}([C]) \mid v \in V \,\}$
        \STATE Define $s(v,k):=\left|\{u\in N(v)\mid c(u)=k\}\right|$
    	\FOR{$v\in V$}
            \STATE$c(v)\gets$ sample $k\in [C]$ with probability $\propto \exp(-\epsilon s(v,k)/2)$ \COMMENT{Exponential Mechanism}
    	\ENDFOR
    \end{algorithmic}
    \label{alg:dp-exp}
\end{algorithm}

We now present a simplified explanation of Algorithm~\ref{alg:dp-exp}. First, we calculate a differentially private estimate of the max degree $\tilde{\Delta}=\Delta+Lap(\frac{1}{\epsilon})$. Each vertex is then tentatively colored uniformly at random from a palette size of $C=\frac{\tilde{\Delta}}{\log n}$. Then we resample the color of each vertex iteratively using the exponential mechanism with the score function $-s(v,k)$ for the color $k$, where $s(v,k)$ is the number of neighbors of $v$ with color $k$.

\subsubsection{Privacy Analysis.} 

\begin{theorem}
Algorithm~\ref{alg:dp-exp} (mechanism $M_{Unctr}$) is $3\epsilon$-differentially private.
\end{theorem}

\begin{proof}
Given the coloring procedure outlined above, we are interested in bounding
\[
    \max_{O \in \C} \max_{G \sim G'} \frac{\Pr[M(G) \in O]}{\Pr[M(G') \in O]} \leq \exp(3\epsilon)
\]
for all possible sets of colorings \(O\) and all edge-neighboring graphs \(G\sim G'\). It is sufficient to consider only individual colorings (ie only consider sets where $|O|=1$).

Let $G\sim G'$ be neighboring graphs, such that $E(G) = E(G') + (u,u')$. Let $\mathbf{c} = (c_1, c_2, \ldots, c_n)$ be a fixed coloring with palette size $C_O$. Let $C(G)$ denote the palette size from Line 2 of Algorithm~\ref{alg:dp-exp}. We note that any processing order in Line 5 and any tentative coloring in Line 3 (assuming the same noisy palette size) are equally probable for both graphs since these do not depend on the structure of the graph. Therefore, it suffices to work with a fixed processing order $v_1, v_2,\ldots, v_n$ and tentative coloring $c_0$ .

The probability that $M(G)$ returns $c$ can be broken down conditioning on the palette size $a$ and the initial coloring $c_0$:
\begin{align}
    &\Pr[M(G)=\mathbf{c}]\\ = &\sum_{a=C_O}^\infty \sum_{c_0}\Pr[C(G)=a]\Pr[c_0]\times \Pr[M(G)=\mathbf{c}|C(G)=a, c_0]
    \label{eq:tot_prob}
\end{align}
Since the tentative coloring $c_0$ is equally probable for both graphs, it is sufficient to show that $\Pr[C(G)=a] \times \Pr[M(G)=\mathbf{c}|C(G)=a]$ differs by at most a factor of $\exp(3\epsilon)$ between $G$ and $G'$. Since $C(G)$ is obtained using Laplace mechanism, for any possible palette size$a$, $\Pr[C(G)=a]/\Pr[C(G')=a]\leq e^\epsilon$. Therefore, conditioning on any fixed palette size $a$, the privacy loss contributed by Lines 1-2 is at most $e^\epsilon$. The second term can be further broken down into:
\begin{align}
  &\Pr[M(G) = \mathbf{c}|C(G)=a]\\ =& \prod_{i=1}^{n} \Pr_G[c(v_i) = c_i| c(v_1) = c_1, \ldots, c(v_{i-1}) = c_{i-1}, C(G)=a]
    \label{eq:cond_chain}
\end{align}
We then analyze the $i^{th}$ term of Equation~\ref{eq:cond_chain} above as
\begin{align*}
&\Pr_G[c(v_i) = c_i|c(v_1) = c_1, \ldots, c(v_{i-1}) = c_{i-1}, C(G)=a] \\
=& \frac{\exp(-\epsilon s_G(v_i, c_i)/2)}{\sum_{j\in [C]}\exp(-\epsilon s_G(v_i, j)/2)}    
\end{align*}

Observe that the utility function $s_G(v_i, j)$ is determined by the conditioned colors of vertices $1, 2, \ldots, i-1$, and the fixed tentative coloring $c_0$ of vertices $i+1, \ldots, n$ (all colors are fixed).

For vertices $v_i\neq u, u'$, $s_G(v_i, j)$, the above probability is the same for graphs $G, G'$, as $v_i$ has the same neighbors, all of which have the same color in both graphs. For $u, u'$, the vertices gain/lose one neighbor of color $c(u'), c(u)$ respectively, so their utility function $s_G(u, c(u')),s_{G}(u',c(u))$ may differ by at most $1$ (leaving all other colors unaffected). Hence, the above probability differs by a factor of $e^\epsilon$ each, by definition of the exponential mechanism.

Therefore, the overall difference in $\Pr[C(G)=a]\Pr[M(G)=\mathbf{c}|C(G)=a]$ between $G, G'$ is by a factor of at most $\exp(3\epsilon)$.

\end{proof}

\subsubsection{Utility Analysis}

\begin{lemma}
    Fix a vertex $v$, and a color $c$ that appears $k$ times among $v$'s neighbors, the probability of vertex $v$ picking color c is at most $\frac{C}{\Delta e^{(k-\frac{\Delta}{C})\epsilon/2}}$, where $C$ is the palette size.
    \label{lemma:initial}
\end{lemma}

\begin{proof}
    By the definition of the Exponential mechanism, the probability that vertex $v$ picks $c$ is $\frac{e^{\epsilon(\Delta-k)/2}}{\sum_c e^{\epsilon(\Delta-s(v,c))/2}}$. We argue that the denominator is minimized when all colors are used evenly. In the regime where colors are distributed evenly, i.e. every color has score $s(v,c)=\frac{\Delta}{c}$, the normalization constant simplifies to $Z_{even}=Ce^{\epsilon(\Delta-\Delta/C)/2}$. Compare this to the regime where every neighbor has the same color, and all other colors are unused, we have $Z_{concentrated}=1+(C-1)e^{\epsilon\Delta/2}$. We note that $\frac{Z_{concentrated}}{Z_{even}}\approx \frac{(C-1)e^{\epsilon\Delta/2}}{Ce^{\epsilon(\Delta-\Delta/C)/2}} = \frac{(C-1)e^{\epsilon\Delta/2C}}{C}$. Given that $C=O(\frac{\Delta}{\log n})$, we have that $e^{\epsilon\Delta/2C}=e^{\epsilon \log n/2} = n^{\epsilon/2}$. Thus, an even distribution gives a much smaller normalization constant. Plugging this back in, $p\le \frac{e^{\epsilon(\Delta-k)/2}}{C(e^{-\epsilon(\Delta-\Delta/C)/2})}$. This simplifies to $p \le \frac{e^{\epsilon(\Delta-k)/2}}{C e^{\epsilon(\Delta-\frac{\Delta}{C})/2}}=\frac{1}{C e^{(k-\frac{\Delta}{C})\epsilon/2}}$.

\end{proof}

\begin{theorem}
Algorithm~\ref{alg:dp-exp} achieves $O(\log n/\epsilon+d)$-defectiveness for $d$-inductive graphs.
\end{theorem}

\begin{proof}
Using Lemma~\ref{lemma:initial}, we can achieve the following bound on defectiveness. 
The probability that a vertex chooses color \(c\) if \(k\) of its neighbors are already colored \(c\) 
is at most 
\[
    \frac{1}{C e^{(k - \Delta / c)\epsilon}}.
\]
Fix a vertex $v$, there are at most $\Delta/k$ potential colors with defectiveness $k$, we want to ensure
\begin{align*}
    &\Pr[\text{v chooses $c$ with $k$ defectiveness}]\\
    = & \frac{\Delta}{k} \times 
      \frac{1}{C e^{(k - \Delta / C)\epsilon / 2}}
    \le 
    \frac{1}{n^2},    
\end{align*}

so that we can apply a union bound on all vertices.

Assume our palette size is \(C = \frac{\Delta}{\log n}\). 
Plugging in this value and solving for \(k\) gives
\[
    k = O\!\left(\frac{\log n}{\epsilon}\right)
\]
as the maximum defectiveness.

However, this bound does not account for later neighbors 
that may resample into the same permanent color after a vertex has resampled. 
Given a \(d\)-inductive graph, there exists an ordering 
\(v_1, v_2, \ldots, v_n\) such that every 
\(v_i \in V(G)\) has at most \(d\) neighbors in 
\(\{v_1, v_2, \ldots, v_{i-1}\}\). 
If we sample the vertices in the reverse of this ordering, 
then only \(d\)-neighbors may resample into the same permanent color 
after the initial resampling. 
Thus, we have our bound.
\end{proof}

Notice that the above result has one large limitation: the need for a d-inductive ordering.
This from the fact that it is difficult the bound the number of neighbors that resample into a color after a vertex has resampled into a permanent color.
We would like an algorithm that generalizes to all graphs, not $d$-inductive ones. 

One way to control the number of later neighbors that are resampled is by only sampling vertices that experience a "bad" initial sampling. 
Thus, we propose Algorithm~\ref{alg:dp-exp-threshold}, which selectively resample vertices which violate a noisy threshold, and therefore controls the number of neighbors that contribute the defectiveness after initial resampling. 

\subsection{Controlled re-sampling algorithm}

Let $def(v)$ be the defectiveness of $v$ (the number of neighbors that share its color). We shall now design and analyze the second algorithm and its guaranteed defectiveness.

\begin{algorithm}
	\caption{Private coloring by controlled resampling mechanism $M_{Control}$\\
	\textbf{Input: } Graph $G=(V,E)$, privacy parameter $\epsilon$\\
	\textbf{Output: } Coloring $c:V \to [C]$}
	\begin{algorithmic}[1]
	\STATE $\tilde{\Delta}\gets\Delta_G+\text{Lap}(\frac{1}{\epsilon})$
	\STATE $C\gets \frac{\tilde{\Delta}}{\log |V|}$
    \STATE $c \gets \{\, v \mapsto \text{Uniform}([C]) \mid v \in V \,\}$
 	\STATE \label{algl:thresh}$\tilde{T}\gets 5\log n +\frac{2\log n + \log \Delta}{\epsilon}$
    \STATE Define $s(v,k):=\left|\{u\in N(v)\mid c(u)=k\}\right|$
    \STATE Define $\text{def}(v):=s(v,c(v))$
    \STATE $V'\gets\emptyset$
	\FOR{$v\in V$}
		\STATE $\tilde{d}(v)=\text{def}(v)+\text{Lap}(\frac{1}{\epsilon})$
		\IF{$\tilde{d}(v) >\tilde{T}$}
			\STATE $V'\gets V'\cup \{v\}$
            \STATE \label{algl:recolThresh}$c'(v)\gets$ sample $k\in [C]$ with probability $\propto \exp(-\epsilon s(v,k))$
		\ENDIF
	\ENDFOR
    \STATE Assign final coloring $\forall v\in V', c(v)\gets c'(v)$
	\end{algorithmic}
	\label{alg:dp-exp-threshold}
\end{algorithm}
\begin{theorem}
	\label{thrm:exp-priv}
	Algorithm~\ref{alg:dp-exp-threshold} ($M_{control}$) is $5\epsilon$-differentially private.
\end{theorem}

\begin{proof}
The proof of privacy for Algorithm~\ref{alg:dp-exp-threshold} can be shown by the composition of the privacy guarantees of all steps.


Note that Algorithm~\ref{alg:dp-exp-threshold} applies the same Exponential mechanism as Algorithm~\ref{alg:dp-exp} in parallel across a subset of the vertices (line~\ref{algl:recolThresh}).
The only additional source of randomness is when determine if each vertex requires recoloring (line~\ref{algl:thresh}). Let $q_i$ be the indicator that vertex $v_i$ requires recoloring (if $\tilde{d}(v_i) > \tilde{T}$).

Fix any two neighboring graphs $G$ and $G'$, (with $(u,v)$ and $(u',v')$ being the differed edge in $G$ and $G'$ respectively), an identical initial coloring over these graphs $c_1, \cdots, c_n$.

For $s_i \in \{0,1\}$, Note that $Pr[q(v_i)=s_i]$ is identical in $G$ and $G'$ for all vertex $v_i\notin \{u, v\}$, since $v_i$ has the same neighbors in $G$ and $G'$.
For $u$ and $v$,  $Pr[q(u)=s_i] \le e^\epsilon Pr[q(u')=s_i]$ by the property of the Laplace mechanism (the addition or removal of a single neighbor can change $\text{def}(u)$ by at most $1$).
Thus, this mechanism is $2\epsilon$-differentially private.

Combining this with the privacy loss of Algorithm \ref{alg:dp-exp} via composition, we show that Algorithm \ref{alg:dp-exp-threshold} is $5\epsilon$-differentially private.
\end{proof}

On the other hand, the guarantee for the defectiveness by the algorithm is much more involved.
To do so, we shall first make use of the following lemma, that states the upper bound of the probabilities that vertex are recolored.

\begin{lemma}
	Each vertex $v\in V$ requires recoloring with probability at most $p=\frac{2 \log n}{\Delta n^2}$
	\label{lemma:reco}
\end{lemma}

\begin{proof}
Observe that the noisy defectiveness $\tilde{d}(v)$ contains two terms: the number of neighbors that were initialized to its color and the Laplacian noise generated by the Laplace mechanism.
The threshold can also be broken down into two corresponding parts, 
\begin{align*}
5\log n, \frac{2\log n + \log \Delta}{\epsilon}
\end{align*}
Hence $v$ will only be recolored if at least one of the terms of $\tilde{d}(v)$ exceeds its corresponding term in the threshold, which we shall show occurs with probability at most $\frac{\log n}{\Delta n^2}$ for each (union bound over the two events gives the desired $\frac{2\log n}{\Delta n^2}$).

For an arbitrary vertex v, Let $V_1 $ be the number of neighbors that were initialized to the same color as v.
Since the initial coloring is chosen uniformly, the expected number of neighbors that were initialized as the same color is $\mathbb{E}[V_1]= \Delta/c=\log n$. Using Chernoff bounds gives 
	\begin{align*}
		\Pr[V_1\geq (1+\delta)\log n] &\leq \exp \left( -\frac{\delta^2\log n}  {2+\delta} \right) \\
										&\leq  \frac{\log n}{\Delta n^2}
	\end{align*}

    Setting $\delta =5$, we get $\Pr[V_1\geq 5 \log n] \leq \exp \left( -25\log n/7 \right) \le \frac{\log n}{\Delta n^2}$ as we are working in the regime where $\log n < \Delta$.
    Thus our threshold is set at 
    \begin{align}
    5 \log n
    \label{eq:thresh}
    \end{align}
	By the tail bound for Laplace noise with scale \(1/\epsilon\), we have
    \begin{align*}
    \Pr[Z \ge t]
        = \frac{1}{2}\exp(-\epsilon t).
    \end{align*}
    Setting this probability to be at most \(\frac{\log n}{\Delta n^2}\), it suffices to choose \(t\) such that
    \begin{align*}
    \frac{1}{2}\exp(-\epsilon t)
        &\leq \frac{\log n}{\Delta n^2}.
    \end{align*}
    Solving for \(t\), we obtain
    \begin{align*}
    t
        &\geq \frac{2\log n + \log \Delta - \log(2\log n)}{\epsilon}.
    \end{align*}
    For simplicity, we set the threshold at
    \begin{align}
    \frac{2\log n + \log \Delta}{\epsilon}.
    \label{eq:thresh}
    \end{align}

    Union bounding over the the initial defectiveness exceeding \eqref{eq:thresh} and the Laplace noise exceeding $\frac{2\log n+\log \Delta}{\epsilon}$ returns $\frac{2\log n}{\Delta n^2}$.

\end{proof}


We prove the following defectiveness of our algorithm using Lemma~\ref{lemma:reco} to bound the number of neighbors that recolor:
\begin{theorem}
	\label{thrm:exp-util}
	Algorithm \ref{alg:dp-exp-threshold} will produce a coloring with defectiveness
\[
\max \left\{
\begin{aligned}
&5\log n + \frac{2\log n + \log \Delta}{\epsilon}  + \frac{2\log n}{\epsilon},
\\
&(1+\tfrac{4}{\epsilon})\log n
+ \tfrac{2}{\epsilon}\log \Delta
\end{aligned}
\right\}
+ \log n
\]

	with probability at least $1 - 3/n$.
\end{theorem}

\begin{proof}
	We begin by noting that if a vertex $v$ has defectiveness $def(v)\geq 5\log n +\frac{2\log n + \log \Delta}{\epsilon}+\frac{2\log n}{\epsilon}$, it will be recolored unless the Laplace noise exceeds $\frac{2\log n}{\epsilon}$, which occurs with probability at most $1/n^2$ by the definition of Laplacian distribution.

	On the other hand, if a vertex is recolored, then using a result from the Exponential mechanism, with probability $1-e^{-t}$, it will choose an option with the utility score (which is negative of defectiveness):
	\[
		OPT - \frac{2\Delta u}{\epsilon}(\ln |R|+t),
	\]
	where $OPT$ is the score of the optimal choice ($-\log n$ here), and $\Delta_u$ is the sensitivity of the score function (in this case it's $1$), $|R|$ is the number of choices (palette size, $c=\Delta/\log n$).
    Choosing $t=2\log n$ means a defectiveness of $(1+\frac{4}{\epsilon})\log n +\frac{2}{\epsilon}\log \Delta$ except with probability $1/n^2$.

	Hence $v$ will have the ``defectiveness'' at most 
    \begin{align*}
\max \left(
\begin{aligned}
&5\log n
+ \frac{2\log n + \log \Delta}{\epsilon}+\frac{2\log n}{\epsilon}
\\
&(1+\tfrac{4}{\epsilon})\log n
+ \tfrac{2}{\epsilon}\log \Delta
\end{aligned}
\right)
\end{align*}
    with probability at least $1-2/n^2$ (here ``defectiveness'' does not take into account $v$'s neighbors may be recolored, which we shall consider next).

	However, $v$'s neighbors may be recolored as well.
    Model each neighbor $i$ as a random variable that takes value $1$ with probability $\frac{2\log n}{\Delta n^2}$ and $0$ otherwise.
    Define the number of neighbors that recolor as $S=\sum_{i=1}^\Delta X_i$.
    Using Markov bound with $\mathbb{E}[S] =\frac{2\log n}{n^2}$ we get
    \begin{align*}
		\Pr[S \geq \log n] &\leq \frac{2 \log n}{n^2 \log n}
        \\
										 &= 2/n^2
	\end{align*}

    Thus, at most $\log n$ neighbors of $v$ get recolored unless with probability $2/n^2$.
	Union bounding over all three events: (1) the vertex being above the threshold but no recolored, (2) the exponential mechanism giving a bad coloring, and (3) the number of neighbors recoloring is large, gives that a vertex v will have defectiveness
    \[
\max \left(
\begin{aligned}
& 5\log n + \frac{2\log n +\log \Delta}{\epsilon}+\frac{2\log n}{\epsilon},
\\
&(1+\tfrac{4}{\epsilon})\log n
+ \tfrac{2}{\epsilon}\log \Delta
\end{aligned}
\right)
+ \log n
\]
    except with probability $\frac{4}{n^2}$, and applying union bound over all $n$ vertices, results in the aforementioned max defectiveness across the graph except with probability $4/n$ and the Theorem follows.
\end{proof}

\begin{remark}
    Algorithm~\ref{alg:dp-exp-threshold} has stronger theoretical guarantees than Algorithm~\ref{alg:dp-exp} as it provides the same asymptotic guarantees while generalizing to non $d$-inductive graphs. It works best when there is an ample privacy budget, as shown in the theoretical bounds we have shown above. However, when under a small $\epsilon$, random initialization theoretically outperforms this algorithm, while using a much tighter privacy budget. The choice of algorithm would be a balance between privacy and defectiveness.
\end{remark}


\section{Experiments}
\label{sec:expr}
\begin{figure*}[ht]
    \centering
    \includegraphics[width=0.9\textwidth]{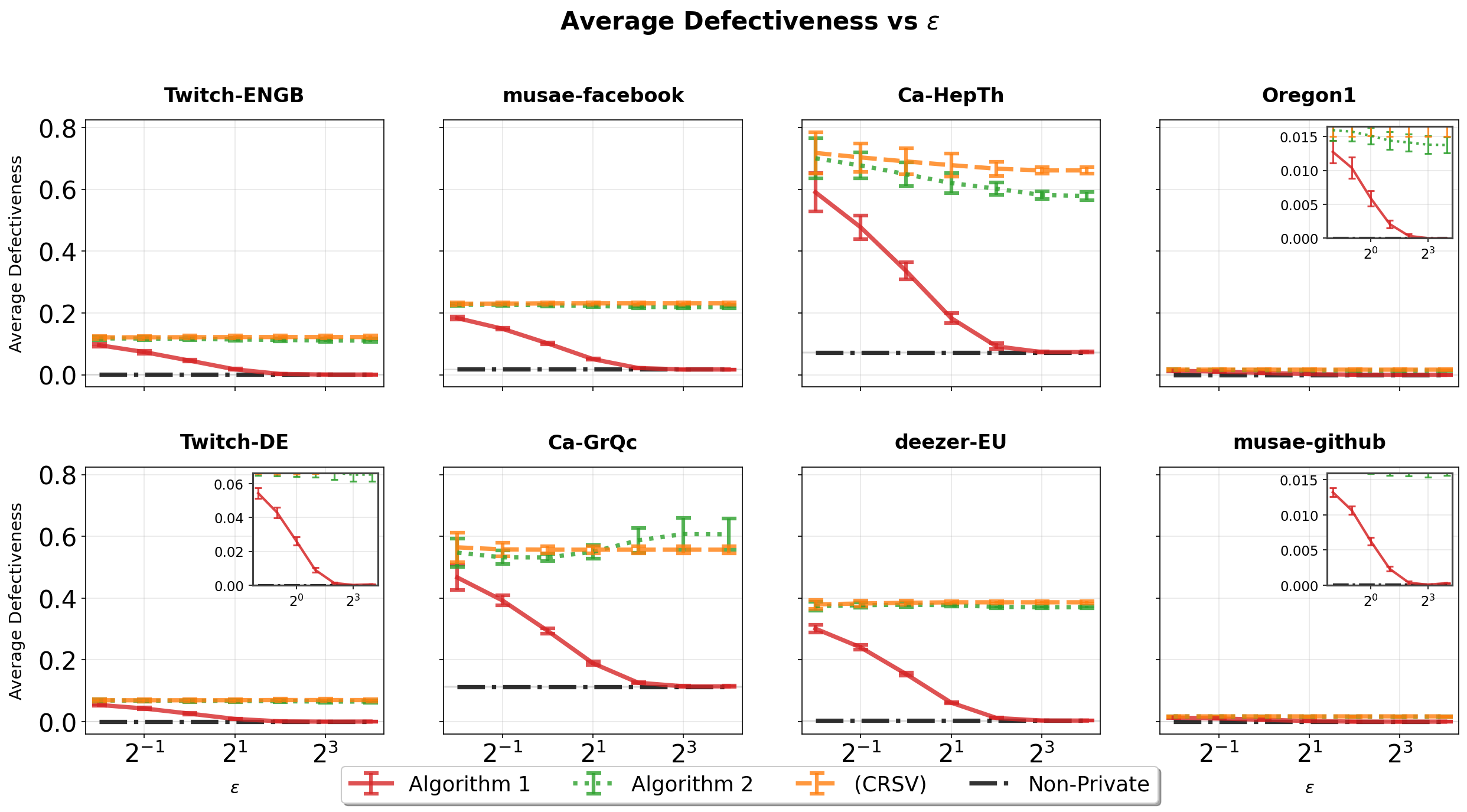}
    \caption{Average defectiveness as a function of $\epsilon$ on SNAP networks. For certain graphs where the lines are indistinguishable, we have provided a zoomed in figure in the top right corner with scaled axes.}
    \label{fig:avg-def-snap}
\end{figure*}

\begin{figure*}[ht]
    \centering
    \includegraphics[width=0.9\textwidth]{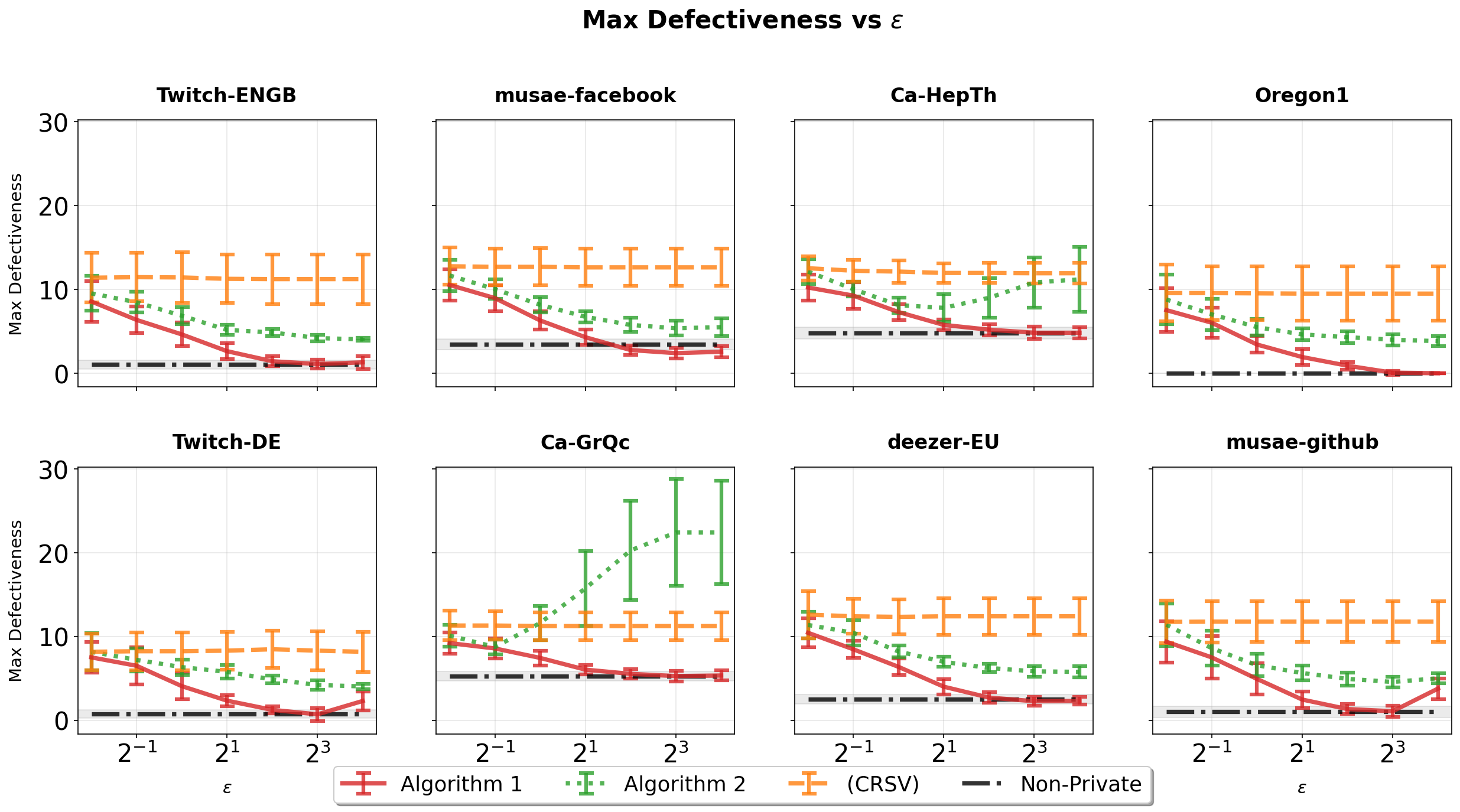}

    \caption{Maximum  defectiveness as a function of $\epsilon$ on SNAP networks.}
    \label{fig:max-def-snap}
\end{figure*}

\subsection{Setups and baselines. } We implement and evaluate several baseline algorithms, including a simple, non-private greedy coloring (Non-private) and the algorithm described by~\cite{christiansen2024private} (\textsc{CRSV}). The non-private algorithm iterates through each vertex and greedily picks the color least used by its neighbors. 
For our experiments, we calculate palette size $|C|=\frac{\Delta+Lap(\frac{1}{\epsilon})}{\log n}$, and use the same palette size for all four algorithms for fair and consistent comparisons. We take all invocations of $\log$ to mean the natural log. We also take the ceiling of the above expression when determining palette size. The palette size is recalculated at the start of every repetition, but within each repetition, the same palette size is shared amongst all algorithms.
We test these algorithms across different values of $\epsilon\in\{2^{-2}, 2^{-1},\ldots, 2^4\}$ (high to low privacy) in order to investigate the privacy-utility trade-off in different privacy settings. For clarity, some values are omitted when labeling the axis of the graphs.

\textbf{Metrics.} After coloring, we count the defectiveness  $def$ of each node $v$. For each network $G$, we report two statistics: the \emph{Average Defectiveness} = $\sum_{v\in V(G)}def(v)/|V(G)|$, and the \emph{Maximum Defectiveness} = $\max_{v\in V(G)} def(v)$. We measure the absolute value of the defectiveness for all of our graphs, and all families of graphs share the same y-axis.  

\textbf{Datasets. } We evaluate the performance of Algorithm~\ref{alg:dp-exp} ($M_{Unctr}$), Algorithm~\ref{alg:dp-exp-threshold} ($M_{Control}$),(\textsc{CRSV}), and a simple non-private greedy coloring (Non-private) on both real world and synthetic graphs. The real world graphs are obtained from the SNAP network repository ~\cite{leskovec2016snap} and are listed in Table~\ref{tab:placeholder_label}. 

For synthetic experiments, we consider both Erdős–Rényi random graphs and Barabási-Albert scale-free networks ~\cite{albert2002statistical}. Erdős–Rényi graphs are generated as $G(n,p)$, where $n$ is the number of vertices and $p$ is the probability of an edge existing between any pair of vertices. Barabási-Albert graphs are generated with $n$ vertices and attachment parameter $m$, where each newly added vertex connects to $m$ existing vertices according to preferential attachment. 

Each experiment, i.e., one instance of an algorithm with a specific combination of parameters ($\epsilon$ for private algorithms) and a specific dataset, is repeated $30$ times and reported by the averages of the statistics. The error bars in the graphs represent standard deviation of the metric over 30 runs. Each repetition is also seeded for reproducibility.
We note that the repeat of experiments is purely for the evaluation of performance.
In practice, multiple invocations of private algorithms incur additional privacy costs, scaling linearly to the number of invocations~\cite{dwork:fttcs14} for pure DP.
Note that the only effect of $\epsilon$ for both (\textsc{CRSV}) and (Non-private) is on the palette size, only in Algorithm~\ref{alg:dp-exp} and~\ref{alg:dp-exp-threshold} does the performance of the algorithm depend on $\epsilon$ directly. When using the original threshold value calculating by Algorithm~\ref{alg:dp-exp-threshold}, we also scaled our theoretical threshold of Line 4 by a factor of 0.1 for all SNAP graphs, and a factor of 0.3 for all synthetic graphs. Note that this does not effect the privacy guarantees of the algorithm, but increases vertex recoloring, meaning Theorem \ref{thrm:exp-util} no longer applies. Because our analysis is based on a worst-case scenario (specifically, graphs that are highly dense and exhibit a relatively uniform degree distribution) and the graphs used in our experiments do not satisfy these assumptions, we adjusted the algorithm's recoloring threshold to better reflect the characteristics of the evaluated datasets and improve empirical performance.

\subsection{SNAP Graphs}

We report summary statistics of our SNAP graphs below.

\begin{table}[hb]
    \centering
    \caption{Statistics of tested SNAP networks}
    \begin{tabular}{|c|c|c|c|}
        \hline
        Graph Name & Nodes & Edges & Max Degree \\ \midrule
        Twitch-ENGB & 7126  & 35324 & 720 \\ \hline
        musae-facebook & 22470 & 171002 & 709 \\ \hline
        Ca-HepTh & 9877 & 25998 & 65 \\ \hline
        Oregon1 & 10670 & 22002 & 2312 \\ \hline
        Twitch-DE & 9498 & 153138 & 4259 \\ \hline
        Ca-GrQc & 5242 & 14496 & 81 \\ \hline
        deezer-EU & 28281 & 92752 & 172 \\ \hline
        musae-github & 37700 & 289003 & 9458 \\ \hline
    \end{tabular}%
    \label{tab:placeholder_label}
\end{table}

\begin{figure*}[ht]
    \centering
    \includegraphics[width=0.9\textwidth]{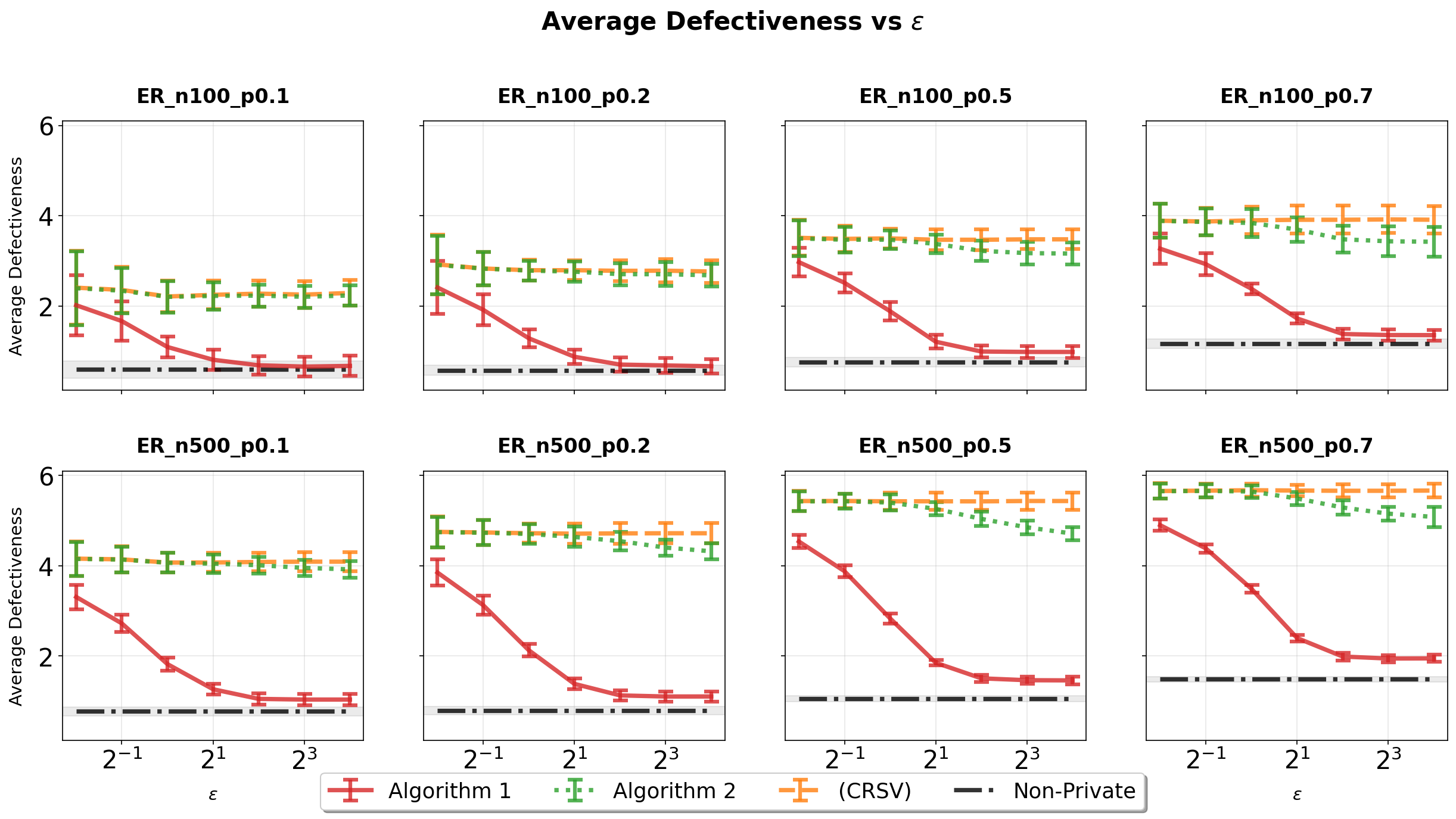}

    \vspace{0.5em}
    \centering
    \includegraphics[width=0.9\textwidth]{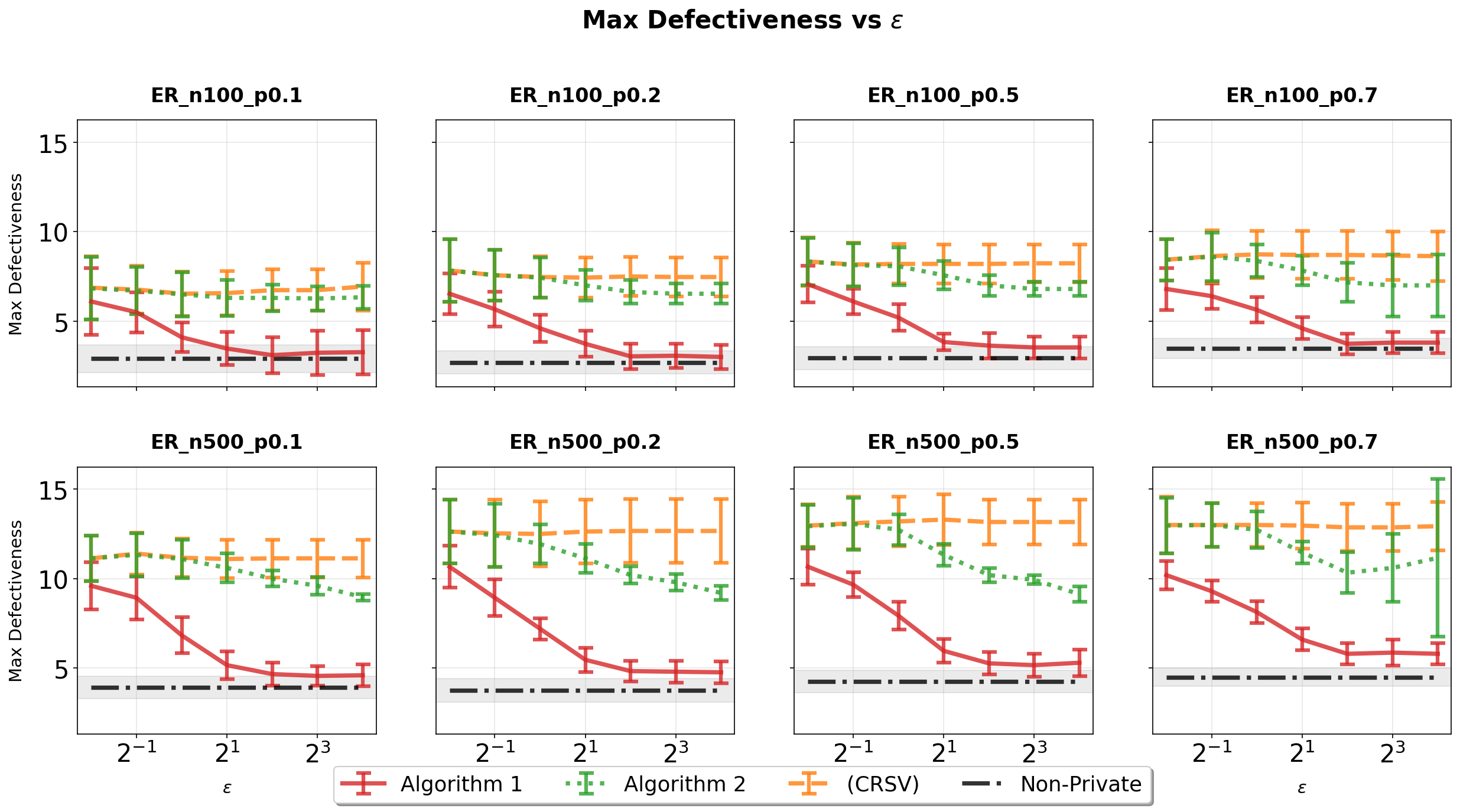}

    \caption{Average (top) and maximum (bottom) defectiveness as a function of $\epsilon$ on $G(n,p)$. }
    \label{fig:defectiveness2}
\end{figure*}

\subsection{SNAP Results}

On the \emph{SNAP networks}, we observe that for Algorithm~\ref{alg:dp-exp}, due to the higher likelihood of worse recoloring by the exponential mechanism, we observe generally larger defectiveness for lower $\epsilon$, and in fact, Algorithm~\ref{alg:dp-exp} tends to approach a greedy non-private baseline as $\epsilon$ increases, thus confirming the tradeoff between defectiveness and privacy. Algorithm~\ref{alg:dp-exp} consistently achieves performance better than both Algorithm~\ref{alg:dp-exp-threshold} and (CRSV) but worse than the greedy non-private baseline at low $\epsilon$, as is expected since we are able to utilize the topology of the graph, but is still constrained by differential privacy. The empirical performance of Algorithm~\ref{alg:dp-exp} suggests that its effectiveness may not only be constrained to d-inductive graphs, but to arbitrary graphs as well. We also observe that Algorithm~\ref{alg:dp-exp-threshold} seems to perform similarly to (CRSV) in terms of average defectiveness. This is because both algorithms begin with the same random initial coloring of the graph. Then, Algorithm~\ref{alg:dp-exp-threshold} selects a subset of high defectiveness nodes to recolor. Since the amount of these nodes is relatively small, especially for graphs which have a small number of very high degree nodes, this does not have a large impact on the average defectiveness over the entire graph. This is also why $\epsilon$ does not seem to impact the average defectiveness in a significant manner.  

For the maximum defectiveness (shown in Figure~\ref{fig:max-def-snap}), we see that excluding Ca-GrQc, Algorithm~\ref{alg:dp-exp-threshold} outperforms CRSV, and Algorithm~\ref{alg:dp-exp} outperforms both. On most graphs, the thresholding mechanism is effective at recoloring highly defective nodes. The high max defectiveness on Ca-GrQc is likely due to it being a denser graph, combined with the scaled down threshold, which facilitates a degenerate case where many neighbors tend to recolor to the same color. Because Algorithm~\ref{alg:dp-exp-threshold} is a one-round parallel recoloring, this does not give the chance for these high defectiveness nodes to fix themselves in later iterations. 

\subsection{Synthetic Graphs}
Next, we test our algorithms on two classes of synthetic graphs, namely Erdős–Rényi random graphs and Barabási-Albert scale-free networks. For Erdős–Rényi random graphs, $n$ denotes the number of nodes and $p$ denotes the probability of an edge between any two given nodes. Figure~\ref{fig:defectiveness2} compares the above algorithms on eight synthetically generated Erdos-Renyi graphs. We keep the same palette size, but we only scale the threshold by a factor of 0.3. 

We run our experiment on eight classes of Erdős–Rényi graphs, whose parameters are listed below.

\begin{table}[hb]
    \centering
    \caption{Erdős–Rényi graph parameters}
    \begin{tabular}{|c|c|}
        \hline
        n & p \\ \hline
        100 & 0.1 \\ \hline
        100 & 0.2 \\ \hline
        100 & 0.5 \\ \hline
        100 & 0.7 \\ \hline
        500 & 0.1 \\ \hline
        500 & 0.2 \\ \hline
        500 & 0.5 \\ \hline
        500 & 0.7 \\ \hline
    \end{tabular}%
    \label{tab:placeholder_label}
\end{table}

We theorize that synthetic graphs will offer more insight into the efficacy of Algorithm~\ref{alg:dp-exp-threshold}. Since the threshold is dependent on the max degree, if there are very few vertices with high degrees, then this leads to very few vertices being resampled. Thus, using Erdos-Renyi creates graphs where all vertices have a similar degree, which allows for more vertices to be resampled, making the effect of the thresholding mechanism easier to observe. 

In Figure~~\ref{fig:defectiveness3}, we also test these algorithms on eight classes of Barabási-Albert scale-free networks using preferential attachment, whose parameters are listed below. 

\begin{table}[H]
    \centering
    \caption{Barabási-Albert graph parameters}
    \begin{tabular}{|c|c|}
        \hline
        m & n \\ \hline
        100 & 2 \\ \hline
        100 & 5 \\ \hline
        500 & 2 \\ \hline
        500 & 5 \\ \hline
        1000 & 2 \\ \hline
        1000 & 5 \\ \hline
        2000 & 3 \\ \hline
        5000 & 2 \\ \hline
    \end{tabular}%
    \label{tab:placeholder_label}
\end{table}

Our results are shown below.

\begin{figure*}[ht]
    \centering
    \includegraphics[width=0.9\textwidth]{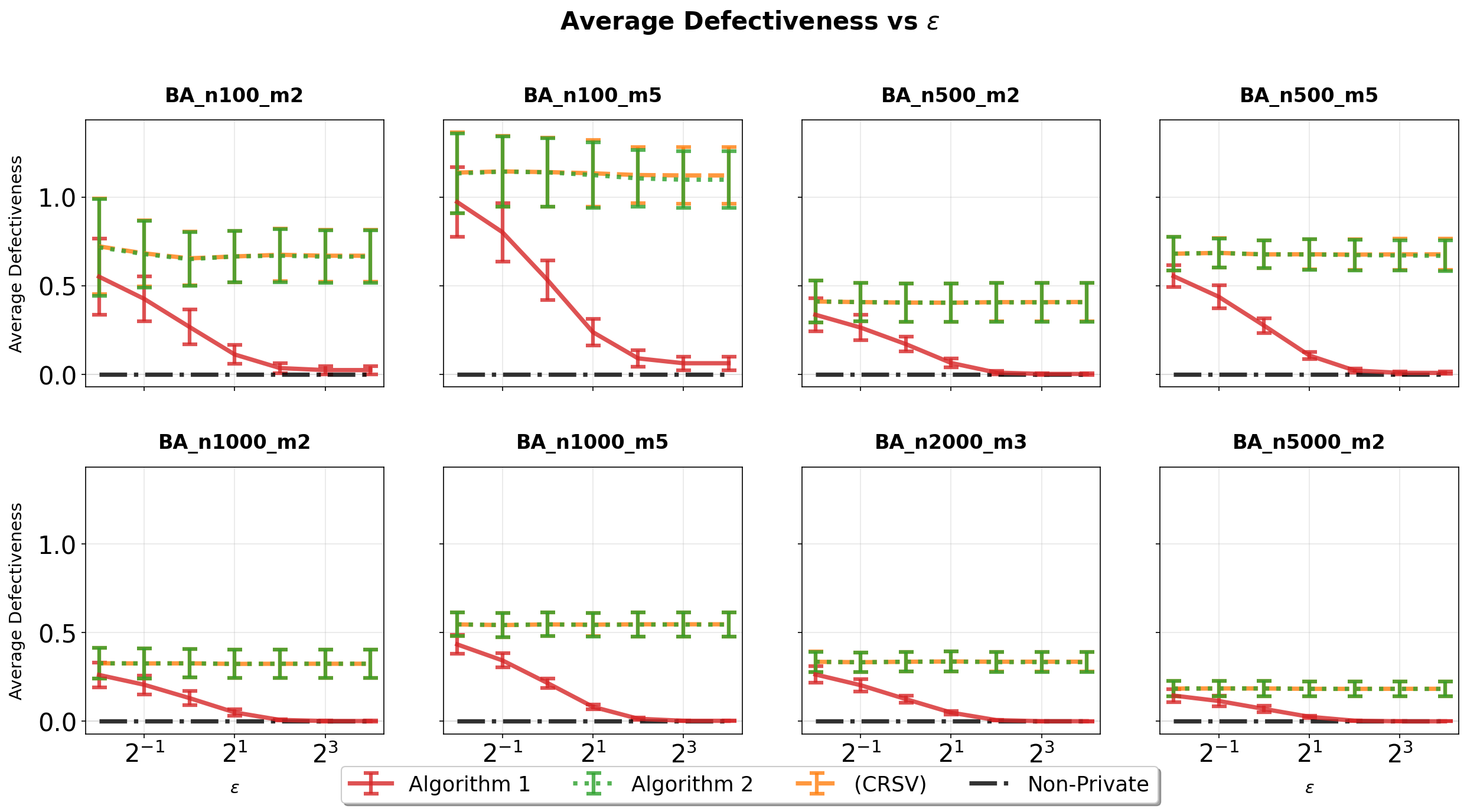}

    \vspace{0.5em}

    \includegraphics[width=0.9\textwidth]{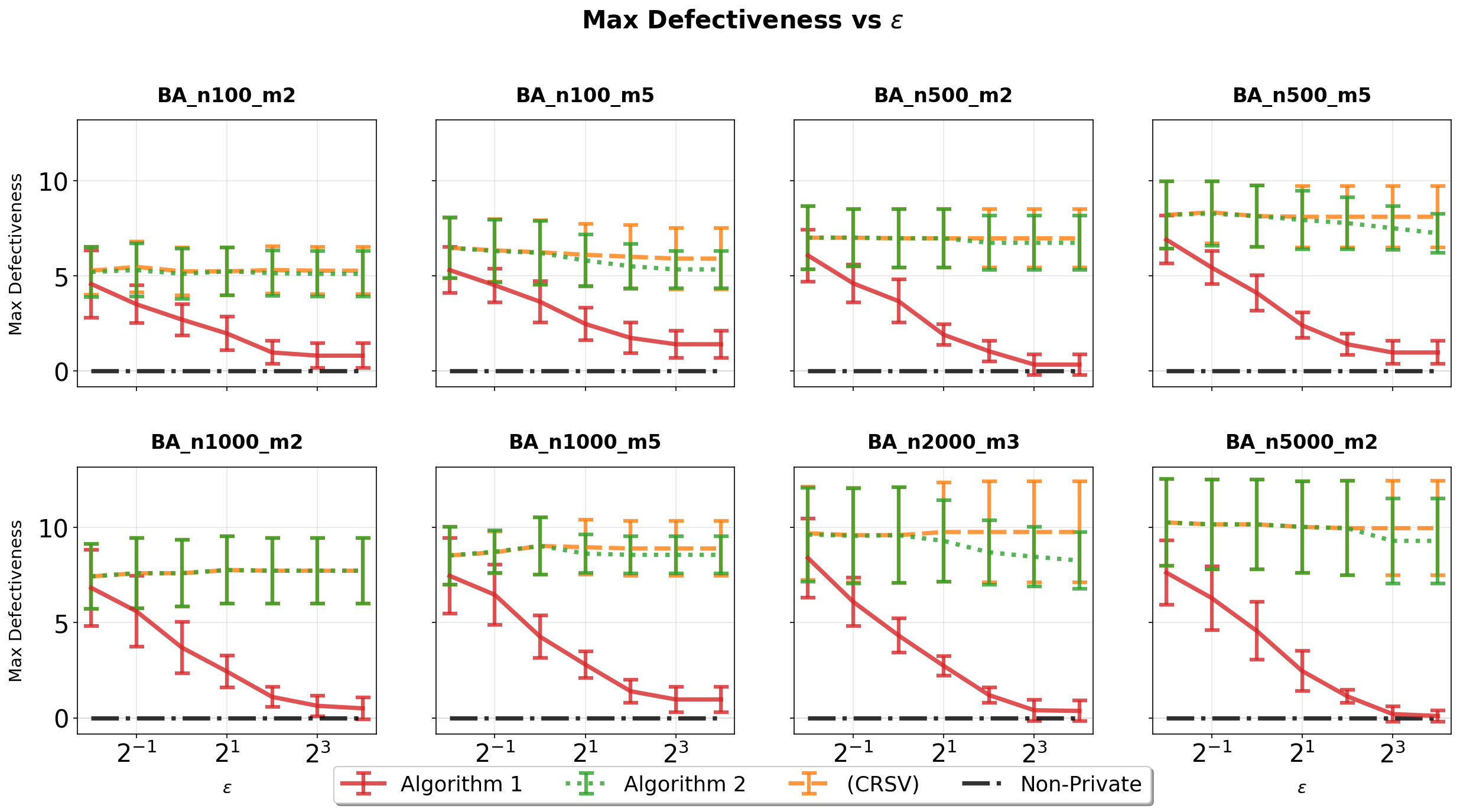}

     \caption{Average (top) and maximum (bottom) defectiveness as a function of $\epsilon$ on Barabási–Albert graphs $\mathrm{BA}(n,m)$. m is a hyperparameter of Barabási–Albert graphs which indicate the number of edges to attach from a new node to existing nodes}
    \label{fig:defectiveness3}
\end{figure*}

\subsection{Synthetic Results}
Our experiments show that across a wide variety of synthetic graphs, Algorithm~\ref{alg:dp-exp} outperforms (CRSV) and Algorithm~\ref{alg:dp-exp-threshold} in terms of both average and max defectiveness (Figures~\ref{fig:defectiveness2} and~\ref{fig:defectiveness3}). Furthermore, we are able to control the utility-privacy tradeoff of the algorithm via the hyperparameter $\epsilon$. This is consistent with our earlier SNAP results as Algorithm~\ref{alg:dp-exp} is able to most effectively utilize the topology of the graph, leading to an overall higher quality coloring. This also reinforces the notion that the theoretical results of Algorithm~\ref{alg:dp-exp} can be extended to arbitrary graphs rather than just d-inductive ones. 

For Erdos-Renyi graphs (Figure~\ref{fig:defectiveness2}), we show that while Algorithm~\ref{alg:dp-exp-threshold} performs slightly better than (CRSV) in terms of average defectiveness, there is a significant improvement in max defectiveness. This is also consistent with our earlier results as Algorithm~\ref{alg:dp-exp-threshold} only resamples a small subset of the vertices, which leads to less impact for average defectiveness. However, the thresholding mechanism ensures that the vertices it does resample are highly defective, which leads to a large improvement in the max defectiveness. In the case of scale-free graphs (Figure~\ref{fig:defectiveness3}), our results are largely consistent with previous results. The average defectiveness is almost identical, this can be attributed to the topology of the graph, as there are a small amount of high-degree vertices which are recolored. However, even this small amount of recoloring leads to an improvement on max defectiveness, over a wide variety of parameters. 

\section{Conclusion}
In this paper, we proposed two novel edge-differentially private graph coloring algorithms. Both algorithms use $O(\frac{\Delta}{\log n})$ colors. Algorithm 1 achieves a maximum defectiveness of $O(\log n + d)$ 
d-inductive graphs, while Algorithm 2 attains a more general bound of 
$O(\log n)$ for arbitrary graphs. These theoretical guarantees are comparable to those of existing differentially private algorithms; however, our empirical evaluation across a wide range of graphs demonstrates that our algorithms achieve substantially lower average and maximum defectiveness while offering a tunable privacy–defectiveness tradeoff. An open question is whether the theoretical guarantees of Algorithm 1 can be generalized to arbitrary graphs, and whether we can achieve better max defectiveness with a larger palette size.



\bibliographystyle{plain}
\bibliography{refs}


\end{document}